\documentclass[preprint,12pt]{elsarticle}

\usepackage{amssymb}
\usepackage{pgf}
\usepackage{tikz}
\usepackage{amssymb,amsmath,amsthm}
\usepackage{multirow}
\usepackage{graphicx}
\usepackage{caption}
\usepackage{hyperref}

\newtheorem{thm}{Theorem}

\theoremstyle{definition}

\newtheorem{rem}{Remark}
\theoremstyle{plain}

\theoremstyle{remark}

\renewcommand{\P}{\ensuremath{\mathbb{P}}}
\newcommand{\e}{e}

\newcommand{\argmax}{\mathop{\mathrm{arg~max}}}
\DeclareMathOperator{\E}{\mathbb{E}}

\newcounter{tmp-counter}

\journal{Performance Evaluation}

\begin{document}

\begin{frontmatter}



\title{Packet Skipping and Network Coding for Delay-Sensitive Network Communication}

\author[label1]{Marc Aoun}
\author[label2]{Paul Beekhuizen}
\author[label3]{Antonios Argyriou}
\author[label1]{Dee Denteneer}
\author[label1]{Peter van der Stok}

 \address[label1]{Philips Research, Eindhoven, 5656AE, The Netherlands}
 \address[label2]{Eurandom, TU/e, Eindhoven, 5656AE, The Netherlands}
 \address[label3]{Department of Computer and Communication Engineering, University of Thessaly, Volos, 38221, Greece}

\begin{abstract}
We provide an analytical study of the impact of packet skipping and opportunistic network coding on the timely communication of messages through a single network element. In a first step, we consider a single-server queueing system with Poisson arrivals, exponential service times, and a single buffer position. Packets arriving at a network node have a fixed deadline before which they should reach the destination. To preserve server capacity, we introduce a thresholding policy, based on remaining time until deadline expiration, to decide whether to serve a packet or skip its service. The obtained goodput improvement of the system is derived, as well as the operating conditions under which thresholding can enhance performance. Subsequently, we focus our analysis on a system that supports network coding instead of thresholding. We characterize the impact of network coding at a router node on the delivery of packets associated with deadlines. We model the router node as a queueing system where packets arrive from two independent Poisson flows and undergo opportunistic coding operations. We obtain an exact expression for the goodput of the system and study the achievable gain. Finally, we provide an analytical model that considers both network coding and packet skipping, capturing their joint performance. A comparative analysis between the aforementioned approaches is provided.

\end{abstract}


\end{frontmatter}


\section{Introduction}
One way to interpret real-time communication is the assignment of particular importance to information flows in the time domain. In this case, the efficacy of data transfer is not measured in terms of raw communication throughput but it becomes dependent on the time instant the information is delivered. The increased interest in real-time communication comes at a time when the Internet becomes a carrier of more time-sensitive information and a new generation of systems such as wireless sensor networks are deployed for the monitoring of time-critical physical environments. High data rate applications like video-on-demand and high quality interactive video communications also require timeliness guarantees.

In this paper, we investigate the communication of real-time messages through a queueing theory -based analysis. We aim at modeling the achievable timeliness benefits when packet skipping and network coding are employed at a router node. We initially model a network node, where data packets have to be serviced (transmitted) within a certain deadline, as a single-server queue. The queueing system is assumed to have an overwrite-buffer with space for one packet; when a new packet bearing more recent data arrives to the queueing system, the buffered packet, if any, is discarded and the new packet is stored in the buffer. Including the buffer space at the server, the system represents an M/M/1/2 queue with deadlines.

The model we consider provides significant insights into the impact of local decisions on the timeliness performance of simple networks. In particular, we show that the real-time goodput of the queueing system is improved when the system employs a policy where buffered packets are only served if the remaining time until their deadline (i.e. their lead-time) is larger than a certain threshold $\theta$. The intuition behind this policy is that server capacity is not wasted on packets that have a low probability of being served before their deadline expires. By skipping service of the currently stored packet, and waiting for a new packet arrival instead, goodput benefits are obtained because a newly arriving packet has a higher probability of meeting its deadline.

Under the assumption that the service times of the single-server system are exponentially distributed, we derive an exact analytic expression for the goodput of the system with this thresholding policy. The exact goodput expression allows us to compute the optimal threshold, which gives insight into the maximum performance gain, as well as the parameter ranges for which the threshold policy provides effective performance benefits. 
In the second part of the paper, the M/M/1/2 queueing system is extended with the ability to algebraically code packets; with network coding, two data packets are transformed into one by a simple algebraic XOR-operation. The coded packet is transmitted and the original native packets can be retrieved at their respective destination nodes through a similar XOR-operation~\cite{ahlswede:network-coding}. The idea of network coding has gained a lot of interest because of its potential to radically affect the way networks operate. When packets with timeliness restrictions undergo network coding, key QoS parameters might be improved but also compromised. Therefore, additional insight into the performance of systems with coding capability in the presence of timeliness requirements is needed. More specifically, it is critical to understand what is the impact of coding on packet delays and the interplay with the system goodput. In this part of the paper, we model a router node to which packets from two flows arrive and must be forwarded. We provide a detailed derivation and exact expressions for the stationary real-time goodput of the router node in the presence of network coding. The analytical results obtained in this paper can be related qualitatively to our recent simulation-based study of the network coding and thresholding policy~\cite{aoun-ewsn2011}. Later, we will discuss aspects of these simulation results that can be related to this paper.

The rest of this paper is organized as follows. Section~\ref{section:deadlines} provides background information on service models and the different types of deadlines that can be associated with real-time information. Related work is presented in Section~\ref{section:related-works}. The system model adopted in this paper is presented in Section~\ref{section:system-model}. The analysis for packet skipping based on lead-time thresholding is described in Section~\ref{section:analysis-thresholding}. Subsequently, the goodput analysis for a system with network-coded communication is developed in Section~\ref{section:performanceAnalysisNetworkCoding}. Section~\ref{section:joint} looks at the joint modeling of skipping and network coding. Finally, Section~\ref{section:conclusions} presents our conclusions and possible directions for future work.

\section{Service Models and Types of Deadlines}\label{section:deadlines}
In queueing systems where packets are associated with deadlines, a service must usually either \emph{begin} or \emph{end} before a specific point in time. In the first case, services are usually non-preemptive (i.e., a service that has begun must be completed), while in the latter case, services are usually preemptive (i.e., a service that has begun can be aborted). In our case, the deadlines are to the end of service, but services are non-preemptive; indeed, the information contained in a packet needs to be completely transmitted before the receiver node can decode the packet and use the information, thus the end-of-service nature of the deadlines. Additionally, serving a packet in the context of wireless network communication consists of possibly contending for wireless channel access, and subsequent transmission of the packet bits. Due to the layered communication architecture adopted in wireless networks, little control can be exerted on the channel contention mechanism once it is activated. The same holds for the physical transmission once it is started.

Kruk et al.~\cite{kruk_edf_heavy_traffic_networks} furthermore consider three types of deadlines: Hard, firm, and soft. With hard deadlines, a single packet that is not served before its deadline makes all other packets useless as well. With this type of deadlines all packets must thus be served before their deadlines, which requires a worst-case analysis. With firm deadlines, every packet that is not served before its deadline is useless, whereas with soft deadlines lateness is permitted and packets can still be used afterwards. We consider the case of firm deadlines, a natural choice that arises when looking at back-end applications requirements in systems such as sensor networks, where conclusions should be based on non-expired data but deadline misses can be tolerated.

\section{Related Work} \label{section:related-works}
Barrer~\cite{Barrer_indifferent_clerks,Barrer_ordered_service} was one of the first to analyze a queue with deadlines. He studied an $M/M/1$ queue with deterministic deadlines. This analysis was extended to systems with state-dependent arrival and service rates, and more general deadlines by Brandt and Brandt~\cite{Brandt_MMs,Brandt_MMs_deadlines} and by Movaghar~\cite{movaghar_beginnings_fifo}. These studies all deal with non-preemptive systems with deadlines to service beginnings, but in~\cite{movaghar_end_fifo}, Movaghar considers deadlines to service endings with preemption. Movaghar and Kargahi~\cite{Kargahi_movaghar,Kargahi_movaghar_edf} have devised an approximation for an $M/M/1$ queue with the earliest-deadline-first (EDF) discipline, which is known to stochastically maximise the fraction of packets served before their deadline (see, e.g., \cite{panwar_towsley_wolf}).

Lehoczky~\cite{Lehoczky_mm1} analyzed an $M/M/1$ queue with deadlines to service endings and the EDF-policy in heavy traffic. He argued that, since the deadlines of all stored packets have to be taken into account, this queue gives rise to a Markov process on a statespace of infinite dimension. Lehoczky shows that the Markov process collapses to a tractable one-dimensional process in heavy traffic. Lehozcky later used these results to analyze control policies in~\cite{Lehoczky_control}, and extended his analysis to Jackson networks in~\cite{Lehoczky_jackson}. Doytchinov et al.~\cite{doytchinov_edf_heavy_traffic} extended this analysis to a $GI/G/1$ queue with deadlines, and Kruk et al.~\cite{kruk_edf_heavy_traffic_networks} to networks of such queues.

Delay sensitive traffic in the presence of network coding was studied by Parag et al. in~\cite{parag08}. The authors adopted a statistical QoS measure that expresses the decay rate of the buffer at the middle node in a butterfly network. This router node is the bottleneck in the butterfly case, which explains the reasons why the authors focus on its buffer behavior. Although this metric is enough for approximating the delivered QoS per flow, it does not express it directly in terms of the achieved packet error rate and goodput. Shah et al.~\cite{shah07} start with the goal of minimizing the backlog of coded packets at receiving nodes. They design an online algorithm so that the linear packet combinations that are generated, are chosen in such a way that their actual span excludes any linear combination that is already known to the receivers.

Eryilmaz et al.~\cite{eryilmaz06} studied the delay benefits of network coding in wireless multicast and multiple unicast scenarios. They presented a model that considers only a single-hop transmission and the random coding across packets from the same flow (intra-session network coding). Online network coding and delay minimization was more recently presented in~\cite{barros09a}. A precise model is not presented in that work, although a simple delay analysis for the wireless channel with Bernoulli erasures is performed. Another interesting modeling work can be found in~\cite{wu09} where the authors used stochastic network calculus for calculating the throughput in a coded butterfly network. Nevertheless, the model does not cover delay and timeliness aspects. A recent work from Goseling et al.~\cite{goseling09} aims at modeling the performance of coded queueing systems. Instead of focusing on providing exact expressions, the authors provide bounds on the performance of a coded tandem queuing network with two independent flows. Nevertheless, the extension of this work to packets with deadline requirements is not straightforward.

\section{System Model} \label{section:system-model}
For the analysis of packet skipping, we consider the single server queueing system depicted in Fig.~\ref{fig: model_sr}.
Packets arrive according to a Poisson process with parameter $\lambda$ and have exponentially distributed service times with parameter $\mu$.
The system has an `overwrite' buffer with one buffer position (not counting the packet in service, if there is one). If a packet arrives when there is already a packet in the buffer, the latter is overwritten (i.e., replaced) by the newly arriving packet. The result is thus a modified version of an $M/M/1/2$ queueing system (note that, in an $M/M/1/K$ system, the $K$ refers to the number of buffer positions, including the one in service).

The same $M/M/1/K$ system is adopted for the analysis of network coding; Two input processes with parameters $\lambda_1$ and $\lambda_2$ are considered, as depicted in Fig.~\ref{fig: model_nc}. The overall arrival process here is also a Poisson process with parameter $\lambda = \lambda_1 + \lambda_2$.
In the network coding case, when a packet is waiting in the buffer, a second packet arriving to the system is algebraically coded with the waiting packet instead of overwriting it, in case they belong to different flows. Overwriting takes place when waiting packet and arriving packet belong to same flow, or when the buffer is already occupied by a coded packet. With network coding the coded packet is transmitted and the original native packets can be retrieved at their respective destination nodes through a similar algebraic XOR operation~\cite{ahlswede:network-coding}. We provide two scenarios where the above is possible. In the case of multicast transmission if each node has one of the two desired packets, they can both recover the other packet after the reception of the single XOR-coded packet~\cite{ahlswede:network-coding}. Also a similar situation arises in the wireless bidirectional relay network where both nodes want to transmit to each other a packet they already have through an intermediate relay node~\cite{aoun-ewsn2011}. In this case similarly, if the two nodes have packets $a$ and $b$, with the transmission of a single coded packet $a\oplus b$ from the relay they could both decode the opposite packet. The broadcast nature of the wireless medium allows this operation.

A packet service is considered successful if it is \emph{completed} within $d$ time units after the arrival of the packet to the queue, i.e., if a packet arrives at time~$t$, its service must end before its global deadline $t+d$, with $d$ being a fixed relative deadline. Services are non-preemptive, which means that if a packet starts service, its service has to be completed, before service on another packet can be started. We define the goodput of the system as the number of packets that are served successfully, i.e. within the deadline requirement, per unit of time.

In the packet skipping case, we introduce a threshold $\theta$ such that the packet in the buffer may only start service if its lead-time, i.e. the remaining time until its absolute deadline, is larger than $\theta$. If the time until the deadline of the packet becomes smaller than $\theta$, the packet is removed from the buffer. The idea behind this threshold is that if the deadline of a packet will be exceeded shortly, there is a large probability that the packet is not served in time. Thus, the goodput of the system might be improved by waiting for another packet that has a higher probability of being served successfully.

A packet whose deadline has expired is not served, both in the system that employs packet skipping and in the system that employs network coding. A coded packet is considered to be expired if the deadlines of both its underlying native packets are expired. This is implemented by setting the expiration deadline of a coded packet to that of the largest among the absolute deadlines of the underlying native packets that were coded together. Note that the individual deadlines of the native packets are still the ones considered when calculating the goodput. Therefore, a coded packet that completes service before its expiration deadline does not necessarily mean that the underlying native packet with the smallest deadline has been served on-time. Also a packet that arrives to an empty system is always served. Indeed, with fixed deadlines, for packet skipping, there is no gain in \emph{not} serving packets arriving to an empty system and waiting for the next packet. After all, the next packet will have the same deadline, and hence the same probability of being served successfully. Delaying the service of a packet arriving to an empty system might be proved to provide additional benefits in the network coding case, but was not considered in this work due to the additional analysis complexity it would engender.

\begin{figure}[t]
\centering
\begin{tikzpicture}[scale=0.8]
  \draw (0,0) -- (1.5,0) -- (1.5,1) -- (0,1);
  \draw (2.1,0.5) circle(0.5cm);
  \draw[->] (-1,0.5) -- node[above]{$\lambda$} (0,0.5);
  \draw[->] (2.6,0.5) -- node[above]{$\mu$}(3.6,0.5);
  \path[fill] (0.75,0.25) rectangle +(0.5,0.5);
  \path[fill] (1.85,0.25) rectangle +(0.5,0.5);
\end{tikzpicture}
\caption{A schematic representation of the model for a node adopting packet skipping.}
\label{fig: model_sr}
\end{figure}
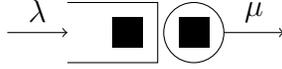

\begin{figure}[t]
\centering
\begin{tikzpicture}[scale=0.8]
  \draw (0,0) -- (1.5,0) -- (1.5,1) -- (0,1);
  \draw (2.1,0.5) circle(0.5cm);
  \draw[->] (-1,0.7) -- node[above]{$\lambda_1$} (0,0.7);
  \draw[->] (-1,0.2) -- node[below]{$\lambda_2$} (0,0.2);
  \draw[->] (2.6,0.5) -- node[above]{$\mu$}(3.6,0.5);
  \path[fill] (0.75,0.25) rectangle +(0.5,0.5);
  \path[fill] (1.85,0.25) rectangle +(0.5,0.5);
\end{tikzpicture}
\caption{A schematic representation of the model with the router exercising network coding.}
\label{fig: model_nc}
\end{figure}
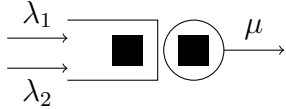

\section{Performance Analysis of Packet Skipping} \label{section:analysis-thresholding}

\subsection{Goodput Analysis}\label{section:renegingGoodputAnalysis}

In this section, we compute the goodput of the system in the presence of packet skipping. It will be convenient to view the goodput as a function of the threshold $\theta$, and we denote it by $\gamma(\theta)$. Furthermore, we denote by $L$ the number of packets present in the system (including the server) immediately before a packet arrival.

A packet arriving to the system will find the latter in one of the following three states:
\begin{itemize}
	\item \textit{Empty (L = 0)}: The server is available, and there is no packet occupying the waiting buffer.
	\item \textit{Busy Server (L = 1)}: The server is servicing a packet, but the waiting buffer is free.
	\item \textit{Full, (L = 2)}: the server is busy and the waiting buffer is occupied. In that case, the arriving packet overwrites the buffered one.
\end{itemize}

\begin{thm} \label{theoremForReneging}
  The goodput $\gamma(\theta)$ is given by
  \begin{equation}\label{eq: gamma}
   \begin{split}
    \gamma(\theta) =& \lambda\: \P(L=0)(1-\e^{-\mu d}) \\
      +&\lambda\:(1-\P(L=0))
            \Big[\frac{\mu}{\lambda+\mu}\left(1-\e^{-(\lambda+\mu)(d-\theta)}\right) \\
            -& \frac{\mu \e^{-\mu d}}{\lambda}\left(1-\e^{-\lambda(d-\theta)}\right)
            \Big],
              \end{split}
  \end{equation}
  where
  \begin{equation}\label{eq: P(L=0)}
    \P(L=0) = 1 - \frac{\lambda^2 + \lambda \mu}{\mu^2 + \lambda \mu \e^{-(\lambda+\mu)(d-\theta)} + \lambda^2 + \lambda \mu}
  \end{equation}
  is the probability that an arbitrary packet arrives at an empty system.
\end{thm}
\begin{proof}
  The goodput $\gamma(\theta)$ is by definition equal to the arrival rate times the probability that the service of an arbitrary packet is successful. To compute this probability, we condition on whether a packet arrives at an empty system or not. Denoting by $B$ the service time that the arriving packet would experience, we have:
  \begin{equation}\label{eq: gamma in P(L=0)}
  \begin{split}
    \gamma(\theta) =& \lambda\:\P(L=0) \P(B \leq d)\\
              +& \lambda\:(1-\P(L=0))\int\limits_0^{d-\theta} \mu \e^{-\mu t} \P(B\leq d-t) \e^{-\lambda t}  dt
  \end{split}
  \end{equation}

  Evaluating the probabilities and the integral indeed yields~\eqref{eq: gamma}.

The rationale behind Eq.~\eqref{eq: gamma in P(L=0)} is as follows: If an arbitrary packet~$P$ arrives at an \emph{empty} system, its service is successful if its service time $B$ is less than $d$. This explains the term $\P(L=0) \P(B \leq d)$. If $P$ arrives at a \emph{non-empty} system ($L = 1$ or $L = 2$), we condition on the length of the remaining (residual) service time~$t$ of the packet in service. The residual service time is exponentially distributed, so its density is $\mu \e^{-\mu t}$. Moreover, given that the residual service time is $t$, the service of~$P$ is successful if the following three conditions are all met: First, the time until the deadline is larger than $\theta$ when the residual service ends, i.e., $d-t\geq \theta$. This is taken into account in the integration region of Eq.~(\ref{eq: gamma in P(L=0)}). Second, the service time is less than $t-d$, which explains the factor $\P(B\leq t-d)$. Third, there were no other arrivals during $t$ time units, which explains the term $\e^{-\lambda t}$.\\

In order to determine $\P(L=0)$, we introduce the concept of cycle time, which we define as the time between two consecutive moments at which the system becomes empty. The cycle time~$C$ consists of two parts: A part during which the system is empty, the `idle time', and a part during which the system serves packets until it becomes empty again, the `busy period'. An illustration of a cycle is presented in Fig.~\ref{fig:cycle}.

The busy period spreads throughout the time that the system is not in a state $L = 0$. It consists of a collection of consecutive states characterized by $L = 1$ or $L = 2$. We shall refer to the case where $L = 2$ during a busy period as \textit{clearance period} ($CP$), which is the amount of time the buffer is continuously occupied. A clearance period ends when the buffer is cleared, i.e. when either a dispatch to the server occurs, or when the stored packet is denied service and removed from the system following the thresholding procedure.

\begin{figure}[t]
\begin{center}
  \includegraphics[scale=0.45]{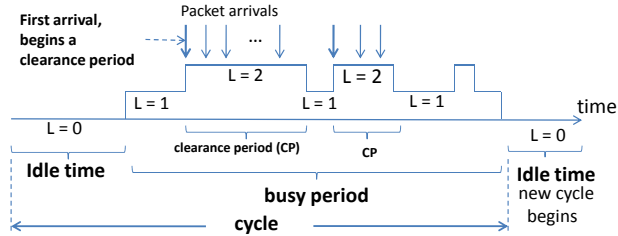}
  \caption{A cycle consisting of idle time and busy period.}
  \label{fig:cycle}
\end{center}
\end{figure}

\smallskip
  The idle time lasts until a packet arrives, so the mean idle time is $1/\lambda$. Denoting the busy period by $BP$, we have:
  \begin{equation*}
    \E[C] = \frac{1}{\lambda} + \E[BP].
  \end{equation*}

  The PASTA-property (Poisson Arrivals See Time Averages, see~\cite{wolff_poisson_1982}) implies that the probability that a packet arrives at an empty system, $P(L=0)$, is given by the probability that the system is empty at an arbitrary time. A standard argument from renewal theory implies that the latter probability is given by the mean idle time divided by the mean cycle time:
  \begin{equation}\label{eq: P(L=0) in E[BP]}
    \P(L=0) = \frac{1/\lambda}{\E[C]} = \frac{1/\lambda}{1/\lambda + \E[BP]}.
  \end{equation}
  The mean busy period is given by the following equation:
  \begin{equation}\label{eq: E[BP] in E[BP]}
    \E[BP] = \frac{1}{\lambda+\mu}+\frac{\mu}{\lambda+\mu}\cdot 0+ \frac{\lambda}{\lambda+\mu}\left(\E[T] + \E[BP]\right),
  \end{equation}
  where $\E[T]$ is the mean buffer clearance period. We will derive $\E[T]$ next.

  Eq.~\eqref{eq: E[BP] in E[BP]} follows from a probabilistic argument: The busy period starts when a packet arrives to an empty system. After this, the expected time until the first event (an arrival or service completion) is $1/(\lambda+\mu)$. Furthermore, with probability $\mu/(\lambda+\mu)$, the first event is a service completion, and the busy period ends after 0 additional time units. With probability $\lambda/(\lambda+\mu)$, the first event is an arrival. In this case, the arriving packet is stored in the buffer and the buffer has to be cleared, which takes in expectation an additional $\E[T]$ time units. Once the buffer has been cleared, the expected time before the system becomes empty is again equal to $\E[BP]$ due to the memorylessness of service times. Evaluating~\eqref{eq: E[BP] in E[BP]} yields
  \begin{equation}\label{eq: E[BP] in E[T]}
    \E[BP] = \frac{1}{\mu} + \frac{\lambda}{\mu} \E[T].
  \end{equation}

  It thus remains to determine $\E[T]$, the expected length of the buffer clearance period. Without loss of generality, we assume that the buffer clearance period starts at time~0. We condition on $t$, the time at which the next event (arrival or service completion) occurs. Because the time until the next event is exponentially distributed with parameter $\lambda + \mu$, we obtain:
  \begin{eqnarray*}
    \E[T] &=& \int\limits_{d-\theta}^{\infty} (\lambda+\mu)\e^{-(\lambda+\mu)t}(d-\theta) dt \\
    &+& \int\limits_0^{d-\theta} \frac{\mu}{\lambda+\mu}(\lambda+\mu)\e^{-(\lambda+\mu)t} t dt \\
    &+&\int\limits_0^{d-\theta} \frac{\lambda}{\lambda+\mu}(\lambda+\mu)\e^{-(\lambda+\mu)t}(t+\E[T])dt.
  \end{eqnarray*}
  The three integrals (from left to right) cover the following three possibilities: First, if there are no events before time $d-\theta$, the buffer is cleared at time $d-\theta$ because the deadline of the buffered packet becomes smaller than $\theta$. Second, if the first event occurs at time $t < d-\theta$, and the first event is a service completion, the buffer is cleared at time~$t$ because the buffered packet enters service. Third, if the first event occurs at time $t<d-\theta$ and is a packet arrival, the arriving packet overwrites the buffered packet, and a new buffer clearance period begins.

  After rewriting and evaluating the integrals (the rightmost two using partial integration), we obtain:
  \begin{equation}\label{eq: E[T]}
    \E[T] = \frac{1-\e^{-(\lambda+\mu)(d-\theta)}}{\mu + \lambda \e^{-(\lambda+\mu)(d-\theta)}}.
  \end{equation}
  Combining Eqs.~\eqref{eq: P(L=0) in E[BP]}, \eqref{eq: E[BP] in E[T]}, and~\eqref{eq: E[T]} yields Eq.~\eqref{eq: P(L=0)}.
\end{proof}

\begin{rem}
If $d\to \infty$, deadlines become irrelevant and two special cases occur for specific values of $\theta$. If $\theta=0$, packets in the buffer are always served, regardless of the time until their deadline. The goodput of the system is thus equal to that in an $M/M/1/2$ queue, namely $\lambda(1-\frac{1}{1+\rho+\rho^2})$, with $\rho = \lambda/\mu$ (see e.g., \cite[Section~5.7]{wolff_intro_to_queueing_theory} or~ \cite[Section~3.6]{Kleinrock_queueing_systems_vol1}). If $\theta=d \to \infty$, packets in the buffer are never served. Only packets arriving to an empty system are. In this case, the goodput is equal to that in an $M/M/1/1$ queue, namely $\lambda(1-\frac{1}{1+\rho})$. By substituting $d$ and $\theta$, these values follow from Eq.~\eqref{eq: gamma} as well.
\end{rem}

\begin{figure}[t]
\centering
\includegraphics[keepaspectratio,width = 0.9\linewidth]{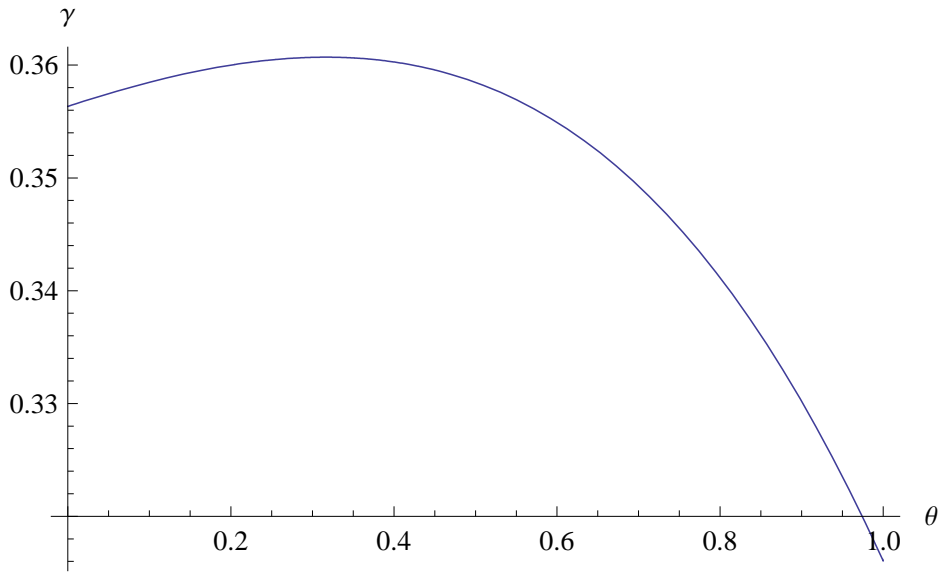}
\caption{The goodput as a function of $\theta$, with $\mu = \lambda = d = 1$.}
\label{fig: throughput mu=lambda=d=1}
\end{figure}

\subsection{Numerical results} \label{section:numerical-results-threshold}

In this section, we study the goodput numerically. In Fig.~\ref{fig: throughput mu=lambda=d=1}, we present $\gamma(\theta)$ for various values of $\theta$, with $\mu = \lambda = d = 1$. We clearly see that the goodput of the system is indeed increased by the threshold~$\theta$, as long as $\theta$ is chosen appropriately.

Having established that an appropriately chosen threshold increases goodput, we study \emph{how much} the goodput can be increased by this threshold. Note that $\gamma(0)$ is the goodput of the system without threshold. We define $\theta^*$ as the threshold that maximizes the relative increase in goodput, i.e., we define
\begin{equation*}
   \theta^* = \argmax\left\{\frac{\gamma(\theta)-\gamma(0)}{\gamma(0)} : 0\leq \theta \leq d\right\}.
\end{equation*}
We define the maximal goodput gain as the maximal relative increase in goodput, i.e., as ($\gamma(\theta^*)-\gamma(0))/\gamma(0)$.

In the sequel, we fix $\mu=1$. We can make this assumption without loss of generality; all parameters are relative to each other, so for any set of parameters, we can scale time in such a way that $\mu=1$ without changing the goodput of the system.

In Fig.~\ref{fig: maximal gain up to 5} we display the maximal goodput gain for $d$ between 0 and~2, and $\lambda$ between 0 and~5. Fig~\ref{fig: maximal gain up to 5} implies that, in this parameter region, we can obtain an increase of up to $15\%$ in goodput by setting the threshold to its optimal value. Furthermore, we see that the relative increase in goodput is maximal if $d \approx 0.3$, but even for larger values of $d$ there can be a goodput enhancement. In addition to this, the maximal goodput gain grows as $\lambda$ grows, so the goodput policy is especially beneficial if the system is overloaded.

\begin{figure}[t]
\begin{center}
\includegraphics[scale=0.85]{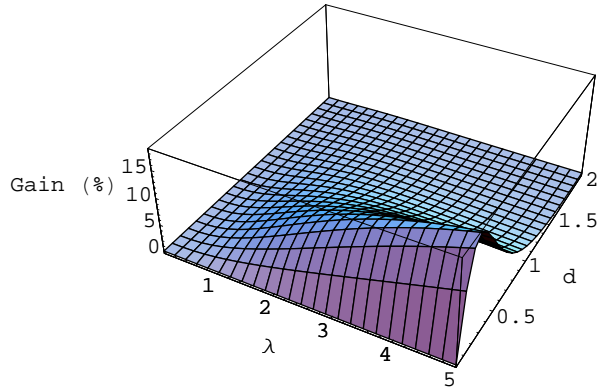}
\caption{The maximal goodput gain for $0\leq \lambda \leq 5$.}
\label{fig: maximal gain up to 5}
\end{center}
\end{figure}


Likewise, in Fig.~\ref{fig: maximal gain up to 10} we show the maximal goodput gain for $\lambda$ up to 10. In this region, the increase in goodput can be as much as $30\%$. Furthermore, the value of $d$ with the largest relative increase in goodput, as well as the range of $d$ for which gain is achieved become smaller. This is conveyed by Fig.~\ref{fig: Goodput different rates}, which plots the gain for three increasing values of $\lambda$, equal to 3, 4 and 6.
The narrower range of $d$ where gain is achieved is explained by a shorter sojourn time of buffered packets that later enter service, due to more frequent buffer overwriting for an increasing arrival rate. The shorter sojourn time reflects itself in a larger lead-time when the buffered packet is assessed, by the packet skipping component, for dispatch to the server. This in turns presents itself in a decrease of the upper value of $d$ for which packet skipping still results in service denials for some packets, and as such, still results in goodput gain.

\begin{figure}[t]
\begin{center}
  \includegraphics[scale=0.9]{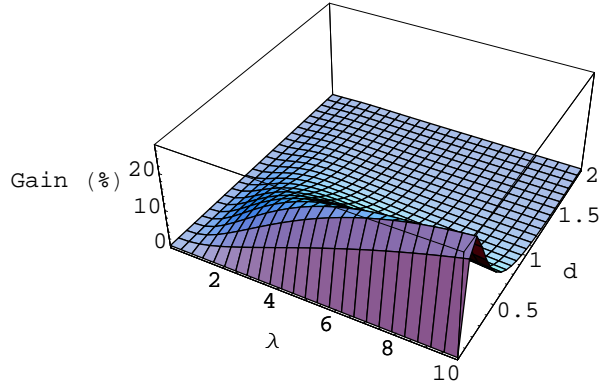}
  \caption{The maximal goodput gain for $0\leq \lambda \leq 10$.}
  \label{fig: maximal gain up to 10}
\end{center}
\end{figure}

\begin{figure}[t]
\begin{center}
  \includegraphics[scale=0.5]{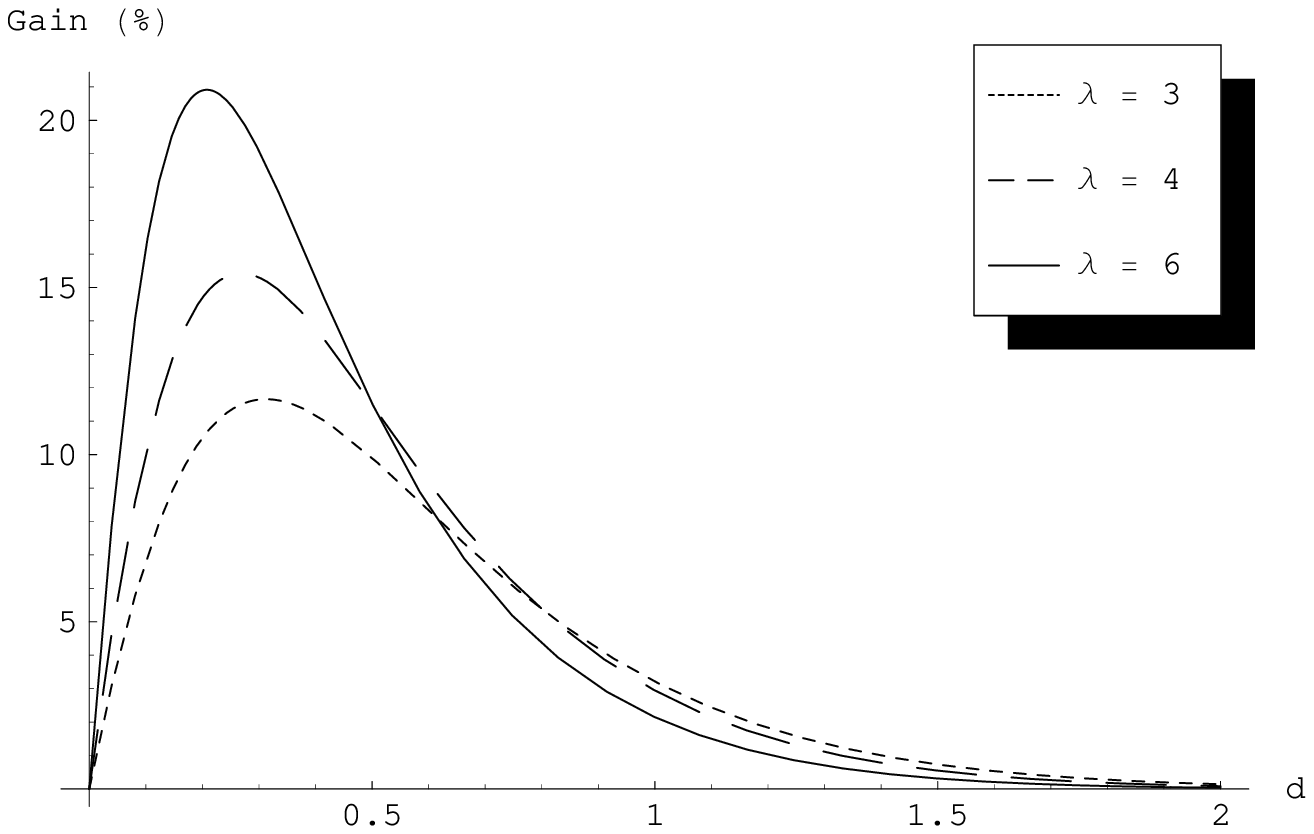}
  \caption{Goodput gain for $\lambda$ values equal to 3, 4 and 6.}
  \label{fig: Goodput different rates}
\end{center}
\end{figure}

Finally, we show values of the gain for specific values of $\lambda$ and $d$ in Table~\ref{tab: maximal gain values}. We clearly see that, for this parameter region, the maximal goodput increase grows as $\lambda$ becomes larger. Furthermore, for a fixed $\lambda$, the maximal gain increases up to a certain value of $d$, and decreases beyond that value.

\begin{table}[t]
\centering
\begin{tabular}{|cc|ccccc|}
\hline
&\multirow{2}{*}{Gain (\%)}
& \multicolumn{5}{c|}{$d$}\\
  && 0.1 & 0.2 & 0.3 & 0.4 & 0.5 \\
  \hline
  \multirow{5}{*}{$\lambda$}
 &0.5 & 0.23 & 0.37 & 0.45 & 0.48 & 0.48 \\
 & 1 & 1.04 & 1.71 & 2.07 & 2.21 & 2.18 \\
 & 2  & 1.036 & 5.90 & 6.91 & 7.05 & 6.64 \\
 & 4 & 10.07 & 14.69 & 15.30 & 13.78 & 11.49 \\
 & 8 & 21.36 & 24.20 & 19.61 & 14.31 & 10.08 \\
 \hline
\end{tabular}
\caption{The maximal goodput gain for specific values of $\lambda$ and $d$.}
\label{tab: maximal gain values}
\end{table}

\section{Performance Analysis of Network Coding} \label{section:performanceAnalysisNetworkCoding}
In this part of the paper we focus our analysis on a router element R, to which two flows of packets arrive and need to be routed. The system model is depicted in Fig.~\ref{fig: model_nc}. The router node is capable of opportunistically coding two native (un-coded) packets belonging to different flows, and transmitting the resulting coded packet instead of two native transmissions. For convenience, we refer to packets of the first flow as type-1 packets, while those of the second flow are type-2 packets. A coded packet is referred to as type-3. The sum of the two arrival rates $\lambda_{1}$ and $\lambda_{2}$ is denoted by $\lambda$.

We shall restrict the analysis to one relative deadline $d$ value for both flows. This restriction neither impacts the quality nor the representativeness of the analysis. It is merely used to simplify the underlying presentation. The extension of our analysis to different deadline distribution is straightforward.

\subsection{Goodput Analysis}
We define the \textit{goodput of a flow} as the number of packets of that flow that are served by the router within their deadline requirement, per unit of time. Given the symmetry of the problem, we will concentrate our analysis on the goodput of type-1 packets. The analysis for packets of type-2 is identical.

The state of the queueing system at the router is fully conveyed through the characterization of the status of its waiting buffer and the status of its server. Two variables will be used to keep track of the system state: $L$, already introduced in Section~\ref{section:renegingGoodputAnalysis}, that defines the current number of packets in the system (including a packet being served), and $N$, that defines the type of the packet residing in the buffer, if any. Knowledge of the type of a packet when it is in service is not required for the analysis.

Variable $L$ can take one of three values: 0, 1, or 2. Note that a coded packet is counted as a single packet in the system. Variable $N$ can take one of three values: 1 (packet in buffer is of type-1), 2 (packet in buffer is of type-2), or 3 (packet in buffer is a coded packet, type-3).

A packet of type-1 arriving to the system can find it in one of five states:
\begin{itemize}
	\item \textit{Empty (L = 0)}: The server is available, and there is no packet occupying the waiting buffer.
	\item \textit{Busy Server (L = 1)}: The server is servicing a packet, but the waiting buffer is free.
	\item \textit{Full, type-1 packet in buffer (L = 2, N = 1)}: the server is busy. The waiting buffer is occupied by a type-1 packet.
	\item \textit{Full, type-2 packet in buffer (L = 2, N = 2)}: the server is busy. The waiting buffer is occupied by a type-2 packet.
	\item \textit{Full, type-3 packet in buffer (L = 2, N = 3)}: the server is busy. The waiting buffer is occupied by a coded packet.
\end{itemize}

\begin{thm} \label{theoremForCoding}
The goodput of type-1 packets is given by the following expression:
\begin{equation}\label{eq:LawTotalProbability}
\gamma_{1} = \lambda_1 \sum_{a,b} \P(L = a, N = b) \P(success|L = a, N = b),
\end{equation}
where $\P(L = a, N = b)$ is the probability that the system is in state $(L = a, N = b)$ and $\P(\mbox{success}|L = a, N = b)$ is the probability that an arbitrary type-1 packet arriving to the system in state (L = a, N = b) meets its deadline.
\end{thm}

\begin{proof}
Follows a similar reasoning as the proof of Theorem~\ref{theoremForReneging}; The goodput $\gamma_{1}$ of type-1 packets is by definition equal to the probability that the service of an arbitrary packet of flow 1 is successful, i.e. meets its deadline. By the PASTA property~\cite{wolff_poisson_1982}, the probability that an arbitrary type-1 packet arrives to the system in state $(L=a,N=b)$ is equal to the steady-state probability that the system is in state $(L=a,N=b)$. The overall service success probability is therefore given by summing, over the entire state space, the probability that the system is in state $(L = a, N = b)$ times the success probability conditioned on the occurrence of this state.
\end{proof}
\smallskip
Finding $\gamma_{1}$ amounts to finding the success probabilities per system state, $\P(\mbox{success}|L = a, N = b)$, and the different state probabilities $\P(L = a, N = b)$. We derive these probabilities in the following two subsections.

\subsection*{Success Probability per System State}

We start by finding the probabilities $\P(\mbox{success}|L = a, N = b)$ that a packet of type-1 arriving to the system when it is in state $(L = a, N = b)$ is served by the router server and completes its service before the expiration of its deadline. The arrival of the packet is considered as time origin. Therefore, the absolute deadline of the packet is equal to its relative deadline.

\smallskip
\subsubsection*{Case 1: Arrival to an Empty System (L = 0)}

A packet arriving to an empty system directly enters service. Its success probability is equal to:
\begin{equation}\label{eq:PsuccessL0}
    \P(success|L  = 0) = (1 - \e^{-\mu d}),
  \end{equation}
which is equal to the probability that the service time it experiences is smaller than its deadline $d$.

\smallskip
\subsubsection*{Case 2: Arrival to a Busy Server (L = 1) or Full System (L = 2) with N = 1 or 3}
Due to the overwrite property of the waiting buffer, an arrival of a type-1 packet to a Busy Server state is identical to an arrival when the system is full and the buffered packet is of type-1 or type-3. Indeed, in all three cases, the newly arriving packet will occupy (or overwrite) the buffer and wait to be serviced. Therefore $\P(\mbox{success}|L = 2, N = 1)$ and $\P(\mbox{success}|L = 2, N = 3)$ are both equal to $\P(\mbox{success}|L = 1)$, given by:


\begin{equation}\label{eq:PsuccessL1}
\begin{split}
\P(success|L = 1) & =  \int\limits_{0}^{d}  \mu\e^{-\mu t} \e^{-\lambda_{1} t} (1+\lambda_{2}t)\e^{-\lambda_{2} t}   \\ &\times(1-\e^{-\mu(d-t)})dt.
\end{split}
\end{equation}

To derive Eq.~\eqref{eq:PsuccessL1}, we first condition on the length of the residual service time \textit{t} of the current packet in the server. Being exponentially distributed, this residual service time has a density equal to $\mu \e^{-\mu t}$. Given that the residual service time is $t$, a buffered packet is successfully serviced if all following conditions are met:
\begin{enumerate}
	\item The server becomes free before the deadline expiration of the buffered packet, i.e., $t\leq d$.
	\item The packet is not overwritten while waiting in the buffer.
	\item The service time of the packet is smaller than the remaining time the packet has until deadline expiration.
\end{enumerate}

Condition 1 is accounted for in the integration region (0 to $d$). Condition 2 is met if and only if, starting from the arrival moment 0 of the buffered packet until the residual service time \textit{t} is completed, no new type-1 packet arrives and at most one arrival of type-2 occurs. Indeed, the first arrival of a type-2 packet will not result in overwriting the buffered packet, since it will be coded with it (resulting in a type-3 packet). Any arrival of type-1 is not tolerated since it overwrites both a type-1 and a type-3 buffered packet. The term $\e^{-\lambda_{1} t}$ gives the probability that no type-1 arrivals occur during t. The term $(1 + \lambda_{2}t)\e^{-\lambda_{2} t}$ gives the probability that at most one type-2 arrival occurs.

Finally, condition 3 is met if the service time experienced by the packet is \textit{at most} equal to the remaining time $(d-t)$ until deadline expiration. The term $(1-\e^{-\mu(d-t)})$ is equal to the probability that the service time is smaller than $d-t$.

\smallskip

\subsubsection{Case 3: Arrival to a Full System with Buffered Type-2 Packet (L = 2, N = 2)}

If the type-1 packet arrives to a system where the waiting buffer is occupied by a type-2 packet, both packets will be coded together, resulting in a type-3 buffered packet. The success probability of the arriving packet is therefore equal to the probability that the resulting type-3 packet is served before the deadline $d$ of the arriving packet is expired. It is given by:

\begin{equation}\label{eq:PsuccessL2N2}
\begin{split}
    \P(success| L = 2, N = 2) &
    = \int\limits_{0}^{d}  \mu\e^{-\mu t} \e^{-\lambda_{1} t} \e^{-\lambda_{2} t}
    \\ &\times (1-\e^{-\mu(d-t)})dt.
    \end{split}
  \end{equation}

Eq.~\eqref{eq:PsuccessL2N2} differs from Eq.~\eqref{eq:PsuccessL1} in that no arrival of type-2 is allowed at all, throughout the whole sojourn time of the coded packet in the buffer. Indeed, any type-2 arrival will overwrite the coded packet. This arrival restriction on type-2 packets is reflected in Eq.~\eqref{eq:PsuccessL2N2} through the term $\e^{-\lambda_{2} t}$.

\subsection*{System State Probabilities}

In order to determine the probabilities $\P(L = a, N = b)$ of finding the system in a particular state $(L = a, N = b)$, we again make use of the notion of cycle time. For additional clarity of the following analysis, the reader is referred to Fig.~\ref{fig:cycle}.

The arrival that causes a state change from $L = 1$ to $L = 2$ (i.e. that starts a clearance period) is referred to as \textit{first arrival}. As a busy period might comprise multiple clearance periods, there can be multiple first arrivals during a busy period. Note that the notion of a type-3 packet as first arrival is abstract in the sense that it does not occur in practice, but is introduced for analysis purposes.

The idle time lasts from the moment the system becomes empty until the arrival of a packet of either type-1 or type-2. Therefore, the mean idle time is equal to $1/(\lambda_{1}+\lambda_{2})$. The mean busy period is the same as that given in Eq.~\eqref{eq: E[BP] in E[T]}, with $\lambda$ equal to the sum of $\lambda_{1}$ and $\lambda_{2}$. Similarly, the expected length of the buffer clearance period, $\E[T]$, is obtained from Eq.~\eqref{eq: E[T]} by setting $\theta$ to 0. The mean cycle time $\E[C]$ is obtained as the sum of the mean idle time and the mean busy period.

As standard argument from renewal theory, the probability that the system is in a certain state is given by the mean time the system spends in this state during a cycle, divided by the mean cycle time. Let \textit{$C_{i}$} be the total time the system has a packet of type \textit{i} in the buffer during a cycle time. In other words, $C_{i}$ is the total amount of time the system is in state $(L = 2, N = i)$ during a cycle.
The system state probabilities are then obtained using:

     \begin{equation}\label{eq:PL0}
    \P(L=0) = \frac{\E[I]}{\E[C]},
    \end{equation}

    \begin{equation}\label{eq:PL2}
    \P(L=2,N=i) = \frac{\E[C_{i}]}{\E[C]}, \text{and}
    \end{equation}

    \begin{equation}\label{eq:PL1}
    \P(L=1) = 1 - P(L=0) - \sum_{i=1}^3 P(L=2,N=i).
    \end{equation}

To find $\E[C_{i}]$, we define \textit{$T_{i,j}$} as the cumulative amount of time that the buffer is occupied by type-i packets during a clearance period, given that the \textit{first arrival} of that clearance period was a type-j packet. The expected value of \textit{$C_{i}$} is given by simply summing up the two different cases of packet types for the first arrival:
    \begin{eqnarray}\label{eq:ECI}
    \E[C_{i}] & = & \sum_{j=1}^2 \frac{\lambda_{j}}{\lambda + \mu}(\E[T_{i,j}] + \E[C_{i}])\nonumber\\
    & = & \sum_{j=1}^2 \frac{\lambda_{j}}{\mu}(\E[T_{i,j}]).
    \end{eqnarray}

It now finally remains to find the expected values of the different $T_{i,j}$. The expected value of $T_{1,1}$ is given by:
\begin{equation}\label{eq:ET11}
\begin{split}
\E[T_{1,1}] & = \int\limits_{0}^{d}  (\lambda+\mu)\e^{-(\lambda+\mu)t}\left[\frac{\mu}{\lambda+\mu} t + \frac{  \lambda_{1}}{\lambda+\mu}\left(E[T_{1,1}]+t\right)\right.\\
  & \left. + \frac{\lambda_{2}}{\lambda+\mu}(E[T_{1,3}]+t) \right] dt + d\e^{-(\lambda+\mu)d}.
\end{split}
\end{equation}
Eq.~\eqref{eq:ET11} conveys the following: After the \emph{first arrival} an event is bound to occur: either a service completion, or a new arrival. We condition on $t$, the time at which the next event occurs. Because of the memorylessness of the system, $t$ has density $(\lambda+\mu) \e^{-(\lambda+\mu)t}$. With probability $\frac{\mu}{\lambda+\mu}$, this event is a service completion, and the clearance period is equal to $t$. With probability $\frac{\lambda_{2}}{\lambda+\mu}$, the first event is a type-2 arrival. Due to memorylessness, this can be seen as the beginning of a new clearance period with a type-3 arrival as \emph{first arrival}, i.e., with expectation $\E[T_{1,3}]$. However, up to time $t$, there was a type-1 packet in the buffer, so the expected \emph{cumulative} amount of time type-1 packets spend in the buffer is given by $t+\E[T_{1,3}]$. With probability $\frac{\lambda_1}{\lambda+\mu}$, the first event is a type-1 arrival. Likewise, this can be seen as a new clearance period with a type-1 arrival as a \emph{first arrival}, i.e., with expectation $\E[T_{1,1}]$ and cumulative amount of time equal to $t+\E[T_{1,1}]$. Finally, with probability $\e^{-(\lambda+\mu)d}$, no event happens prior to the deadline expiration and subsequent removal of the type-1 first arrival. In that case, the buffer would have been occupied by this type-1 packet for an amount of time equal to $d$.

Likewise, the expected values of $T_{1,2}$ and that of $T_{1,3}$ are respectively given by:

    \begin{equation} \label{eq:ET12}
    \begin{split}
    \E[T_{1,2}] =  \int\limits_{0}^{d}  & (\lambda+\mu)\e^{-(\lambda+\mu)t} \\
     & \left[\frac{\lambda_{1}}{\lambda+\mu}E[T_{1,3}]  +  \frac{\lambda_{2}}{\lambda+\mu}E[T_{1,2}]\right] dt,
    \end{split}
    \end{equation}

    \begin{equation}\label{eq:ET13}
    \begin{split}
    \E[T_{1,3}] = \int\limits_{0}^{d} & (\lambda+\mu)\e^{-(\lambda+\mu)t} \\
    & \left[\frac{\lambda_{1}}{\lambda+\mu}E[T_{1,1}] + \frac{\lambda_{2}}{\lambda+\mu}E[T_{1,2}]\right] dt.
    \end{split}
    \end{equation}

   Let $\alpha = 1 - \e^{-(\lambda + \mu)d}$. The upper set of integrals results in the following set of three equations with three unknowns, which can be solved to find $\E[T_{1,1}]$, $\E[T_{1,2}]$ and $\E[T_{1,3}]$:

    \begin{equation} \label{eq:ET11set}
    E[T_{1,1}] = \alpha\left[\frac{1}{\lambda + \mu} + \frac{\lambda_{1}}{\lambda + \mu} E[T_{1,1}] +  \frac{\lambda_{2}}{\lambda + \mu} E[T_{1,3}]\right]
    \end{equation}

    \begin{equation}\label{eq:ET12set}
    E[T_{1,2}] = \alpha\left[\frac{\lambda_{1}}{\lambda + \mu} E[T_{1,3}] +  \frac{\lambda_{2}}{\lambda + \mu} E[T_{1,2}] \right]
    \end{equation}

    \begin{equation}\label{eq:ET13set}
    E[T_{1,3}] = \alpha\left[\frac{\lambda_{1}}{\lambda + \mu}E[T_{1,1}] +  \frac{\lambda_{2}}{\lambda + \mu} E[T_{1,2}] \right]
    \end{equation}

    Following a similar reasoning,

    \begin{equation} \label{eq:ET21set}
    E[T_{2,1}] = \alpha\left[\frac{\lambda_{1}}{\lambda + \mu} E[T_{2,1}] +  \frac{\lambda_{2}}{\lambda + \mu} E[T_{2,3}]\right]
    \end{equation}

    \begin{equation} \label{eq:ET22set}
    E[T_{2,2}] = \alpha\left[\frac{1}{\lambda + \mu} + \frac{\lambda_{1}}{\lambda + \mu} E[T_{2,3}] +  \frac{\lambda_{2}}{\lambda + \mu} E[T_{2,2}] \right]
    \end{equation}

    \begin{equation} \label{eq:ET23set}
    E[T_{2,3}] = \alpha\left[\frac{\lambda_{1}}{\lambda + \mu}E[T_{2,1}] +  \frac{\lambda_{2}}{\lambda + \mu} E[T_{2,2}] \right]
    \end{equation}

 and

     \begin{equation} \label{eq:ET31set}
    E[T_{3,1}] = \alpha\left[\frac{\lambda_{1}}{\lambda + \mu} E[T_{3,1}] +  \frac{\lambda_{2}}{\lambda + \mu} E[T_{3,3}]\right]
    \end{equation}

    \begin{equation}\label{eq:ET32set}
    E[T_{3,2}] = \alpha\left[\frac{\lambda_{1}}{\lambda + \mu} E[T_{3,3}] +  \frac{\lambda_{2}}{\lambda + \mu} E[T_{3,2}] \right]
    \end{equation}

    \begin{equation} \label{eq:ET33set}
    E[T_{3,3}] = \alpha\left[\frac{1}{\lambda + \mu} + \frac{\lambda_{1}}{\lambda + \mu}E[T_{3,1}] +  \frac{\lambda_{2}}{\lambda + \mu} E[T_{3,2}] \right]
    \end{equation}

Once the different $\E[T_{i,j}]$ are found, they are used in Eq.~\eqref{eq:PL2} to find the state probabilities $\P(L=2,N=i)$. Replacing the result of Eqs.~\eqref{eq:PL0},~\eqref{eq:PL2}, and~\eqref{eq:PL1} along with the different conditional success probabilities (Eqs.~\eqref{eq:PsuccessL0}-\eqref{eq:PsuccessL2N2}) into Eq.~\eqref{eq:LawTotalProbability} yields the type-1 goodput of the system, which is a function of $\lambda_{1}$, $\lambda_{2}$, $\mu$ and d. 

\begin{figure}[t]
\begin{center}
  \includegraphics[scale=0.71]{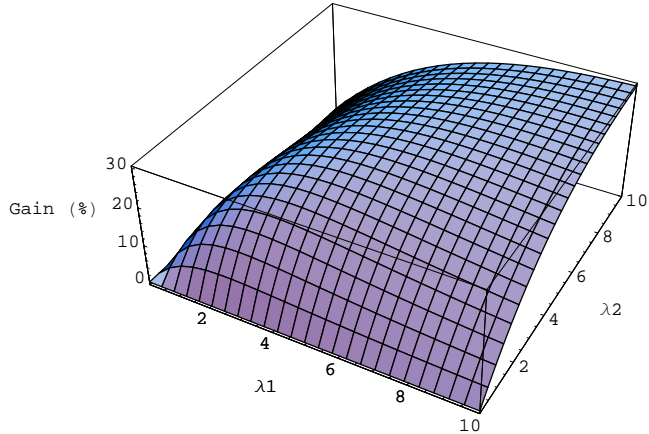}
\end{center}
  \caption{Goodput gain (\%), function of $\lambda_{1}$ and $\lambda_{2}$ when d = 1.}
  \label{fig:gainFunctionLambda1Lambda2}
\end{figure}

\begin{figure}[t]
\begin{center}
  \includegraphics[scale=0.6]{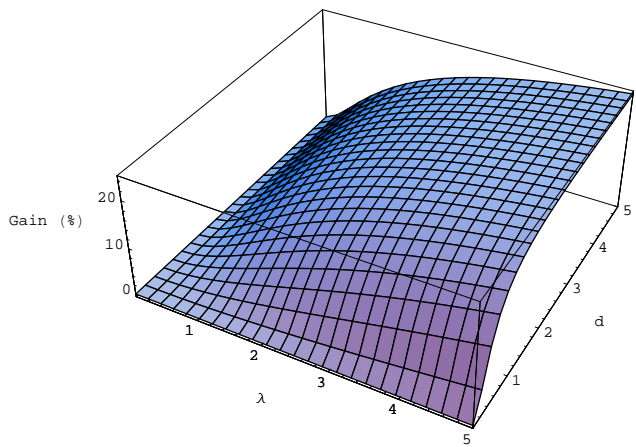}
\end{center}
  \caption{Goodput gain (\%) for values of $\lambda$ between 0 and 5 and values of $d$ between 0 and 5.}
  \label{fig:gainLambda0to5deadline0to5}
\end{figure}

\begin{figure}[t]
\begin{center}
  \includegraphics[scale=0.6]{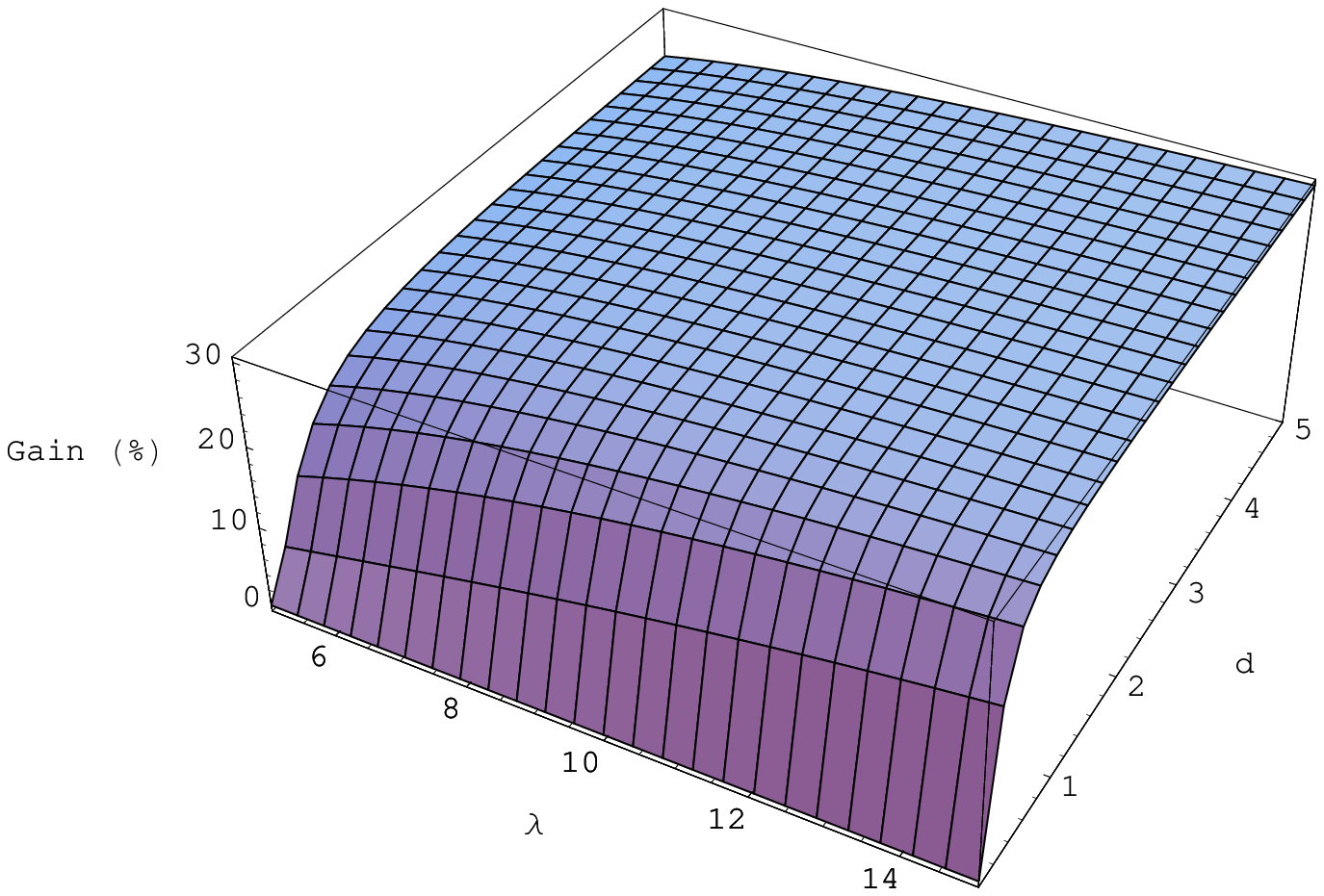}
\end{center}
  \caption{Goodput gain (\%) for values of $\lambda$ between 5 and 15.}
  \label{fig:gainLambda5to15deadline0to5}
\end{figure}

\subsection{Numerical Results} \label{section:numerical-results-coding}
In this section, we study numerically the goodput gain of network coding relative to the no-coding base case. We are particularly interested in finding how much the goodput of real-time packets flows can be increased by applying network coding. We define $\gamma_{base}$ as the total goodput of the router without coding, and $\gamma_{coding}$ as the total goodput of the router with coding. An important remark to be made here is that $\gamma_{base}$ can be found using Eq.~\eqref{eq:LawTotalProbability}, by considering the existence of a single arriving flow instead of two, with parameter $\lambda$ = $\lambda_{1}$ + $\lambda_{2}$. In other words,
 \begin{equation} \label{eq:base case}
  \gamma_{base}\left\{\lambda_{1},\lambda_{2}, \mu, d\right\} = \gamma_{1}\left\{\lambda_{1} + \lambda_{2},0, \mu, d\right\}.
    \end{equation}
We define the goodput gain as the relative increase in goodput, i.e., as $(\gamma_{coding}-\gamma_{base})/\gamma_{base}$.

In the sequel, we fix $\mu$ = 1. Fig. \ref{fig:gainFunctionLambda1Lambda2} provides the goodput gain as a function of the arrival rates $\lambda_{1}$ and $\lambda_{2}$, for $d$ equal to 1. A major conclusion to be drawn here is that the gain provided by network coding is always maximized when $\lambda_{1}$ is equal to $\lambda_{2}$. This result is logical, since equal arrival rates provide the most opportunities for coding.

Fig. \ref{fig:gainLambda0to5deadline0to5} provides the maximal goodput gain for d between 0 and 5, and
$\lambda$ between 0 and 5. Similarly, Fig. \ref{fig:gainLambda5to15deadline0to5} provides the maximal goodput gain for the same range of deadlines and larger values of  $\lambda$, between 5 and 15.

For a fixed arrival rate $\lambda$, the gain increases for an increasing deadline value. A larger deadline provides a higher probability of success for coded packets, hence a positive return on applying network coding. Similarly, for a fixed deadline $d$, the gain increases for an increasing arrival rate; This can be explained by the more efficient buffer usage that network coding provides; In the base case, every arrival overwrites the buffer. The higher the arrival rate, the more frequent such overwritings occur. On the other hand, with network coding, an arrival of type-1 when the buffer is occupied by a type-2 packet, or vice-versa, does not result in an overwriting, but in a coding operation which maintains both packets.
As shown in Fig. \ref{fig:gainLambda0to5deadline0to5}, the gain provided by network coding remains limited in the operating region ($\lambda$ $<$ 2, $d$ $<$ 1), where low arrival rates limit the opportunities of coding, and the short deadlines causes most of the coded packets to miss their timeliness requirement anyway.
Finally, as conveyed in Fig. \ref{fig:gainLambda5to15deadline0to5}, the gain of network coding for high arrival rates and large deadlines can reach up to 30\%, a substantial improvement compared to the base case scenario. This result is important for highlighting the applicability of the proposed approach in more general networks. When the same network coding algorithm was employed for the transmission of real-time packets in a wireless multihop sensor network, the simulation results also showed a maximum performance of 30\% for higher values of the traffic load and higher deadlines~\cite{aoun-ewsn2011}.

A reasoning along the lines of buffer usage efficiency also explains the difference in performance between packet skipping and network coding when the deadline value is increased. As already shown in Fig.~\ref{fig: maximal gain up to 5} and Fig.~\ref{fig: maximal gain up to 10}, the gain of packet skipping is limited to the smallest range of deadlines ($d<1$). In contrast, network coding provides increasing gain for an increasing deadline value, up to reaching a saturation value of 30\%. For higher deadlines, lead-time thresholding basically allows any buffered packet to be serviced, thus providing no improvement in buffer usage. On the other hand, the efficient usage of the buffer is maintained in the case of network coding, since coding operations are deadline-oblivious.

\section{Joint Coding-Skipping Model} \label{section:joint}

\subsection{Goodput Analysis} \label{section:analysis-network-joint}

The following section studies the performance benefits obtained when network coding and packet skipping are concurrently applied. With packet skipping enabled in a network coding setting, instead of removing a waiting packet only once its absolute deadline expires, lead-time thresholding with threshold $\theta$ is applied; if the waiting packet is a native one (a type-1 or type-2 packet), its lead-time is used in the thresholding procedure. If the waiting packet is a coded packet, the largest among the lead-times of the two underlying native packets is used.

The goodput analysis for the joint coding-skipping approach is essentially identical to the analysis of the network coding goodput in Section~\ref{section:performanceAnalysisNetworkCoding}. A number of equations are affected by the introduction of $\theta$; $E[T]$ takes exactly the expression provided in Eq.~\eqref{eq: E[T]}, in comparison to the simple network coding case where $\theta$ was set to 0. The threshold $\theta$ is also introduced in the equations of $\P(success|L = 1)$ and $\P(success| L = 2, N = 2)$, resulting in:

\begin{equation}\label{eq:PsuccessL1New}
\begin{split}
\P(success|L = 1) &= \int\limits_{0}^{d-\theta}  \mu\e^{-\mu t} \e^{-\lambda_{1} t} (1+\lambda_{2}t)\e^{-\lambda_{2} t}   \\ &\times(1-\e^{-\mu(d-t)})dt.
\end{split}
\end{equation}

\begin{equation}\label{eq:PsuccessL2N2New}
\begin{split}
    \P(success| L = 2, N = 2) &= \int\limits_{0}^{d-\theta}  \mu\e^{-\mu t} \e^{-\lambda_{1} t} \e^{-\lambda_{2} t} \\
     & \times (1-\e^{-\mu(d-t)})dt.
\end{split}
  \end{equation}

Note that the equations giving $\P(success|L = 2, N =1)$ and $\P(success|L = 2, N =3)$ remain identical to Eq.~\eqref{eq:PsuccessL1New}.

Finally, the threshold $\theta$ should be taken into account in all $E[T_{i,j}]$ expressions (Eqs.~\eqref{eq:ET11set}-\eqref{eq:ET33set}), by setting $\alpha$ to $1 - \e^{-(\lambda + \mu)(d-\theta)}$ instead of $1 - \e^{-(\lambda + \mu)d}$.

\begin{figure}[t]
\begin{center}
  \includegraphics[scale=0.64]{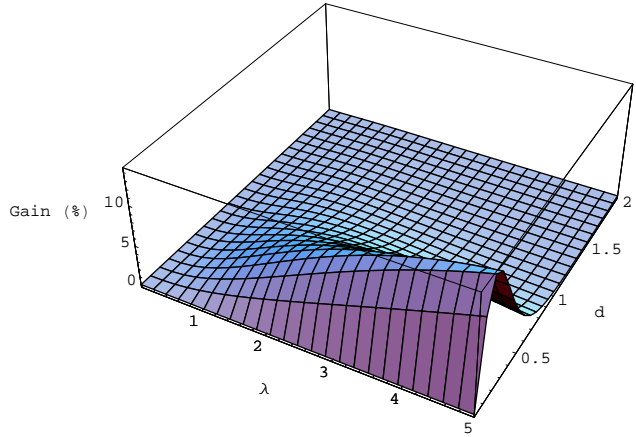}
\end{center}
  \caption{Additional increase in goodput gain (\%), for $\lambda$ between 0 and 5, achieved when adding skipping to network coding.}
  \label{fig:comparisonRngandRngImpactInCombinedLambdaEqual5}
\end{figure}

\begin{figure}[t]
\begin{center}
  \includegraphics[scale=0.64]{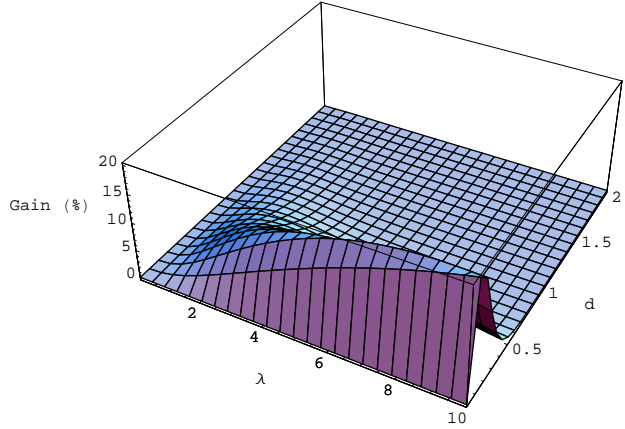}
\end{center}
  \caption{Additional Increase in Goodput gain (\%), for $\lambda$ up to 10, achieved when adding skipping to network coding.}
  \label{fig:comparisonRngandRngImpactInCombinedLambdaEqual10}
\end{figure}

\begin{figure}[t]
\begin{center}
  \includegraphics[scale=0.64]{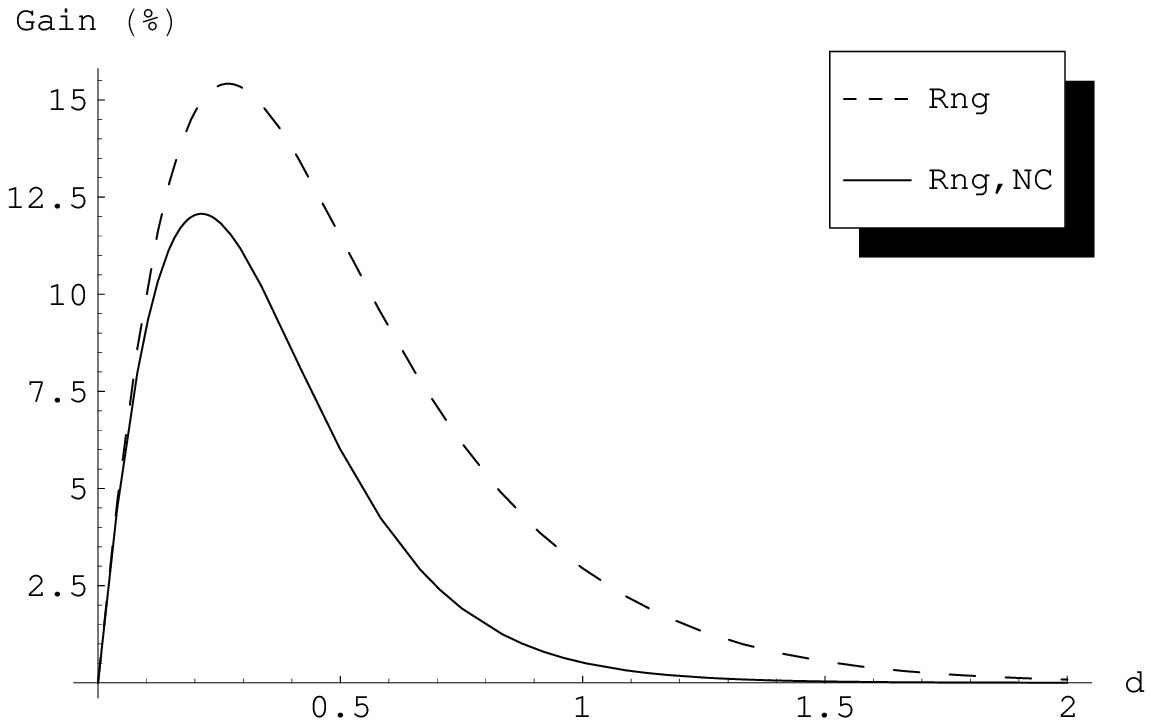}
\end{center}
  \caption{Comparison, for $\lambda = 4$, between gain of skipping and the additional gain it provides when enabled concurrently with network coding.}
  \label{fig:RngVsRngInJointforLambda4}
\end{figure}

\subsection{Numerical Results} \label{section:numerical-results-joint}
Fig.~\ref{fig:comparisonRngandRngImpactInCombinedLambdaEqual5} and Fig.~\ref{fig:comparisonRngandRngImpactInCombinedLambdaEqual10} plot the \textit{additional} gain of combining service skipping with network coding, relative to the gains enjoyed from network coding alone. The impact of packet skipping in a joint coding-skipping approach remains limited to the smallest range of deadlines, in accordance with the results for which only packet skipping was employed. However, we observe that the range where the performance of the joint system is noticeable, is smaller than the one of the packet skipping case. A snapshot of this difference can also be seen in Fig.~\ref{fig:RngVsRngInJointforLambda4}. This behavior is consistent with our simulation results in~\cite{aoun-ewsn2011} where we observed that the additional gain from packet skipping was reaching at most 10-12\% and for the lower regime of deadline values. This behavior puzzled us both in our analysis here and also in our simulations. But careful analysis reveals that when network coding is employed, our system naturally removes opportunities for skipping since many native packets, once coded, are transmitted even if they would have been normally dropped with thresholding. This situation is even more true for higher arrival rates since more packets are coded and they are replaced from the buffer faster. Therefore, the opportunities that occur for efficient use of packet skipping become less as the arrival rate is increased.

%

\section{Conclusions}
\label{section:conclusions}
In this paper, we developed an analytical model that captures the impact of packet skipping and network coding on the delivery of real-time packets. First, we developed a model that considers only packet skipping for packets that have associated deadlines. We showed through the analytical model, that an optimal thresholding policy, when applied to the incoming packets at the router, can lead to goodput gains of up to 30\%.
Subsequently we developed a model for the case where the router employs algebraic coding on the incoming packets. We showed that depending on the operating region, the gain of network coding can reach up to 30\%. The ultimate objective of the paper was a router queueing system model that considers jointly network coding and packet skipping in the presence of per-packet deadlines. We provided exact expressions for the stationary goodput of the system and numerical results of the achievable goodput gain. Employing packet skipping in addition to network coding adds up to 15\% gain over the collected gains obtained from employing network coding exclusively.

In addition to the novel approach that looks at network coding and packet skipping from a real-time perspective, our work intends to be a stepping stone for future research. Among the topics we foresee as interesting are the qualitative and quantitative studies of network coding and packet skipping in more complex network topologies. Furthermore, we will seek ways to translate some of the results and conclusions of this work to heuristic optimization algorithms for real-life wireless networks.


\begin{thebibliography}{100}
\expandafter\ifx\csname url\endcsname\relax
  \def\url#1{\texttt{#1}}\fi
\expandafter\ifx\csname urlprefix\endcsname\relax\def\urlprefix{URL }\fi
\expandafter\ifx\csname href\endcsname\relax
  \def\href#1#2{#2} \def\path#1{#1}\fi

\bibitem{ahlswede:network-coding}
R.~Ahlswede, N.~Cai, S.-Y.~R. Li, R.~W. Yeung, Network information flow, IEEE
  Transactions on Information Theory 46~(4) (2000) 1204--1216.

\bibitem{aoun-ewsn2011}
M.~Aoun, A.~Argyriou, P.~van~der Stok, Performance evaluation of network coding
  service reneging in ieee 802.15.4-based wireless sensor networks, in:
  European conference on Wireless Sensor Netoworks (EWSN), 2011.

\bibitem{kruk_edf_heavy_traffic_networks}
L.~Kruk, J.~Lehoczky, S.~Shreve, S.-N. Yeung, Earliest-deadline-first service
  in heavy-traffic acyclic networks, The Annals of Applied Probability 14~(3)
  (2004) 1306--1352.

\bibitem{Barrer_indifferent_clerks}
D.~Barrer, {Q}ueueing with impatient customers and indifferent clerks,
  Operations Research 5~(5) (1957) 644--649.

\bibitem{Barrer_ordered_service}
D.~Barrer, {Q}ueueing with impatient customers and ordered service, Operations
  Research 5~(5) (1957) 650--656.

\bibitem{Brandt_MMs}
A.~Brandt, M.~Brandt, On the m(n)/m(n)/s queue with impatient calls,
  Performance Evaluation 35~(1-2) (1999) 1--18.

\bibitem{Brandt_MMs_deadlines}
A.~Brandt, M.~Brandt, Asymptotic results and a {M}arkovian approximation for
  the m(n)/m(n)/s+gi system, Queueing Systems 41~(1-2) (2002) 73--94.

\bibitem{movaghar_beginnings_fifo}
A.~Movaghar, On queueing with customer impatience until the beginning of
  service, Queueing Systems 29~(4) (1998) 337--350.

\bibitem{movaghar_end_fifo}
A.~Movaghar, On queueing with customer impatience until the end of service,
  Stochastic Models 22~(1) (2006) 149--173.

\bibitem{Kargahi_movaghar}
M.~Kargahi, A.~Movaghar, Non-preemptive earliest-deadline-first scheduling
  policy: A performance study, in: IEEE International Symposium on Modeling,
  Analysis, and Simulation of Computer and Telecommunication Systems, Georgia,
  ATL, USA, 2005, p. 201-210.

\bibitem{Kargahi_movaghar_edf}
M.~Kargahi, A.~Movaghar, A method for performance analysis of
  earliest-deadline-first scheduling policy, Journal of Supercomputing 37~(2)
  (2006) 197--222.

\bibitem{panwar_towsley_wolf}
S.~Panwar, D.~Towsley, J.~Wolf, Optimal scheduling policies for a class of
  queues with customer deadlines to the beginning of service, Journal of the
  ACM 35~(4) (1988) 832--844.

\bibitem{Lehoczky_mm1}
J.~Lehoczky, Real-time queueing theory, in: IEEE Real-Time Systems Symposium,
  Washington, DC, USA, 1996, p. 186-195.

\bibitem{Lehoczky_control}
J.~Lehoczky, Using real-time queueing theory to control lateness in real-time
  systems, Performance Evaluation Review 25~(1) (1997) 158--168.

\bibitem{Lehoczky_jackson}
J.~Lehoczky, Real-time queueing network theory, in: IEEE Real-time Systems
  Symposium, San Francisco, CA, USA, 1997, p. 58-67.

\bibitem{doytchinov_edf_heavy_traffic}
B.~Doytchinov, J.~Lehoczky, S.~Shreve, Real-time queues in heavy traffic with
  earliest-deadline-first queue discipline, The Annals of Applied Probability
  11~(2) (2001) 332--378.

\bibitem{parag08}
P.~Parag, J.-F. Chamberland, Queuing analysis of a butterfly network, in: IEEE
  International Symposium on Information Theory (ISIT), Toronto, Canada, 2008,
  pp. 672--676.

\bibitem{shah07}
D.~Shah, M.~M\'{e}dard, J.~K. Sundararajan, On queueing in coded networks --
  queue size follows degrees of freedom, in: IEEE Information Theory Workshop
  on Wireless Networks, Solstrand, Norway, 2007, pp. 1--6.

\bibitem{eryilmaz06}
A.~Eryilmaz, A.~Ozdaglar, M.~M\'{e}dard, On delay performance gains from
  network coding, in: Conference on Information Sciences and Systems (CISS),
  Princeton, NJ, USA, 2006, pp. 864--870.

\bibitem{barros09a}
J.~Barros, R.~A. Costa, J.~Widmer, Effective delay control for online network
  coding, in: IEEE International Conference on Computer Communications, Joint
  Conference of the IEEE Computer and Communications Societies (INFOCOM), Rio
  de Janeiro, Brazil, 2009, pp. 208--216.

\bibitem{wu09}
K.~Wu, W.~Jia, Y.~Yuan, Y.~Jiang, Performance modeling of stochastic networks
  with network coding, in: Workshop on Network Coding, Theory, and Applications
  (NetCod), Lausanne, Switzerland, 2009, pp. 6--11.

\bibitem{goseling09}
J.~Goseling, R.~J. Boucherie, J.-K. Ommeren, Energy consumption in coded queues
  for wireless information exchange, in: Workshop on Network Coding, Theory,
  and Applications (NetCod), Lausanne, Switzerland, 2009, pp. 30--35.

\bibitem{wolff_poisson_1982}
R.~Wolff, {P}oisson arrivals see time averages, Operations Research 30~(2)
  (1982) 223--231.

\bibitem{wolff_intro_to_queueing_theory}
R.~Wolff, {I}ntroduction to {Q}ueueing {T}heory, Prentice-Hall, Inc., 1989.

\bibitem{Kleinrock_queueing_systems_vol1}
L.~Kleinrock, Queueing Systems. Volume 1: Theory, John Wiley \& sons, 1975.

\end{thebibliography}
\end{document}